\documentclass[11pt]{article}
\usepackage{amsmath,amsthm,amssymb,color,verbatim,graphicx,fullpage}
\newcommand{\remove}[1]{}
\newtheorem{thm}{Theorem}[section]
\newtheorem{clm}[thm]{Claim}
\newtheorem{lem}[thm]{Lemma}
\newtheorem{define}[thm]{Definition}

\newtheorem{open}[thm]{Open Problem}

\renewcommand{\remove}[1]{}

\newcommand{\eps}{{\varepsilon}}
\renewcommand{\l}{\left}
\renewcommand{\r}{\right}

\newcommand{\comments}[1]{}

\renewcommand{\deg}{\textnormal{deg}}

\newcommand{\Maj}{\textnormal{MAJ}}

\def\F{{\mathbb{F}}}

\newcommand{\R}{\mathbb{R}}
\newcommand{\N}{\mathbb{N}}
\newcommand{\T}{\mathbb{T}}
\newcommand{\E}{\mathbb{E}}
\newcommand{\Z}{\mathbb{Z}}
\newcommand{\U}{\mathbb{U}}

\newcommand{\C}{\mathbb{C}}

\newcommand{\wt}{\mathrm{wt}}

\renewcommand{\Pr}{\mathbf{Pr}}

\newcommand{\MOD}{\mathrm{MOD}}

\def\draft{0}   

\ifnum\draft=1 
    \def\ShowAuthNotes{1}
\else
    \def\ShowAuthNotes{0}
\fi

\ifnum\ShowAuthNotes=0
\newcommand{\authnote}[2]{{ \footnotesize \bf{\color{red}[#1's Note: {\color{blue}#2}]}}}
\else
\newcommand{\authnote}[2]{}
\fi

\begin{document}
\title{Nonclassical polynomials as a barrier to polynomial lower bounds}

\author{
Abhishek Bhowmick\thanks{Department of Computer Science.
The University of Texas at Austin.
\texttt{bhowmick@cs.utexas.edu}. Research supported in part by NSF Grant CCF-1218723.
}
\and
Shachar Lovett \thanks{Department of Computer Science and Engineering. University of California, San Diego. 
\texttt{slovett@ucsd.edu}. Research supported by NSF CAREER award 1350481.}
}

\maketitle

\begin{abstract}
The problem of constructing explicit functions which cannot be approximated by low degree polynomials has been extensively 
studied in computational complexity, motivated by applications in circuit lower bounds, pseudo-randomness, constructions of Ramsey graphs
and locally decodable codes. Still, most of the known lower bounds become trivial for polynomials of super-logarithmic degree. Here, we suggest
a new barrier explaining this phenomenon. We show that many of the existing lower bound proof techniques extend to nonclassical
polynomials, an extension of classical polynomials which arose in higher order Fourier analysis. Moreover, these techniques are 
tight for nonclassical polynomials of logarithmic degree.
\end{abstract}

\section{Introduction}

Polynomials play a fundamental role in computer science with important applications in algorithm design, coding theory, pseudo-randomness, cryptography and complexity theory. They are also instrumental in proving lower bounds, as many lower bounds techniques first reduce the computational model to a computation or an approximation by a low degree polynomial, and then continue to show that certain hard functions cannot be computed or approximated by low degree polynomials. Motivated by these applications, the problem of constructing explicit functions which cannot be computed or approximated (in certain ways) by low degree polynomials has been widely explored in computational complexity. However, most techniques to date apply only to relative low degree polynomials. In this paper, we focus on understanding this phenomenon, when the polynomials are defined over fixed size finite fields. In this regime, many lower
bound techniques become trivial when the degree grows beyond logarithmic in the number of variables. We propose a new barrier explaining the lack of ability to prove strong lower bounds for polynomials of super-logarithmic degree. The barrier is based on \emph{nonclassical polynomials}, an extension of standard (classical) polynomials which arose in higher order Fourier analysis. We show that several existing lower bound techniques extend to nonclassical polynomials, for which the logarithmic degree bound is tight. Hence, to prove stronger lower bounds, one should either focus on techniques which distinguish classical from nonclassical polynomials, or consider functions which are hard also for nonclassical polynomials.

\paragraph{Nonclassical polynomials.}
Nonclassical polynomials were introduced by Tao and Ziegler~\cite{TZ11} in their works on the inverse theorem for the Gowers uniformity norms. To introduce these, it will be beneficial to first consider classical polynomials. Fix a prime finite field $\F_p$, where we consider $p$ to be a constant.
A function $f:\F_p^n \to \F_p$ is a degree $d$ polynomial if it can be written as a linear combination of monomials of degree at most $d$. An equivalent definition is that $f$ is annihilated by taking any $d+1$ directional derivatives. That is, for a direction $h \in \F_p^n$ define the derivative of $f$ in direction $h$ as
$D_h f(x) = f(x+h)-f(x)$. Then, $f$ is a polynomial of degree at most $d$ iff
$$
D_{h_1} \ldots D_{h_{d+1}} f \equiv 0 \qquad \forall h_1,\ldots,h_{d+1} \in \F_p^n.
$$
Nonclassical polynomials extend this definition to a larger class of objects. Let $\T=\R/\Z$ denote the torus. For a function $f:\F_p^n \to \T$, define its directional derivative in direction $h \in \F_p^n$ as before, as
$D_h f(x) = f(x+h)-f(x)$. Then, we define $f$ to be a \emph{nonclassical polynomial of degree at most $d$} if it is annihilated by any $d+1$ derivatives,
$$
D_{h_1} \ldots D_{h_{d+1}} f \equiv 0 \qquad \forall h_1,\ldots,h_{d+1} \in \F_p^n.
$$
While not immediately obvious, the class of nonclassical polynomials contains the classical polynomials. Let $|\cdot|:\F_p \to \{0,\ldots,p-1\} \subset \Z$ denote the natural embedding.
If $f:\F_p^n \to \F_p$ is a classical polynomial of degree $d$ then $|f(x)|/p \pmod{1}$ is a nonclassical polynomial of degree $d$. It turns out that as long as $d<p$, these capture all the nonclassical polynomials. However, for $d \ge p$ nonclassical polynomials strictly extend classical polynomials of the same degree. For example, the following is a nonclassical polynomial of degree $p$:
$$
f(x) = \frac{\sum |x_i|}{p^2} \;.
$$
See Section~\ref{sec:prelim} for more details on nonclassical polynomials.

\paragraph{Correlation bounds for polynomials.}
We first consider the problem of constructing explicit boolean functions which cannot be approximated by low-degree polynomials. For simplicity, we focus on polynomials defined
over $\F_2$, but note that the results below extend to any constant prime finite field. This problem was studied by Razborov~\cite{Raz2} and Smolensky~\cite{Sm} in the context of proving lower bounds for $\mathrm{AC}^0(\oplus)$ circuits (and more generally, bounded depth circuits with modular gates modulo a fixed prime). Consider for example the function $\MOD_3:\{0,1\}^n \to \{0,1\}$, which outputs $1$ if the
sum of the bits is zero modulo $3$, and outputs $0$ otherwise. The probability it outputs $0$ is $2/3$. They showed that low degree polynomials over $\F_2$ cannot improve this significantly. If $f:\F_2^n \to \F_2$ be a polynomial of degree $d$ then
$$
\Pr_{x \in \{0,1\}^n}\left[f(x) = \MOD_3(x)\right] \le \frac{2}{3} + O\left( \frac{d}{\sqrt{n}}\right).
$$
This is sufficient to prove that the $\MOD_3$ function cannot be computed by sub-exponential $\mathrm{AC}^0(\oplus)$ circuits. However, one would like to prove that it cannot even be slightly approximated. Such a result would be a major step towards constructing pseudorandom generators for $\mathrm{AC}^0(\oplus)$ circuits~\cite{Ni,NW}, a well known open problem in circuit complexity. It turns out that the Razborov-Smolensky bound is tight for very large degrees, as there exist polynomials of degree $d=\Omega(\sqrt{n})$ which approximate the $\MOD_3$ function with probability $0.99$, say. However, it seems to be far from tight for $d \ll \sqrt{n}$, which suggests that an alternative proof technique may be needed.

Viola and Wigderson~\cite{VW} proved stronger inapproximability results for degrees $d \ll \log{n}$. These are better described if one considers the correlation of $f$ with the sum of the bits modulo $3$. In the following, let $\omega_3 = \exp(2 \pi i / 3)$ be a cubic root of unity. They showed that if $f:\F_2^n \to \F_2$ is a polynomial of degree $d$ then
$$
\E_{x \in \{0,1\}^n}\left[ (-1)^{f(x)} \omega_3^{x_1+\ldots+x_n} \right] \le 2^{-\Omega(n / 4^d)}.
$$
The technique of~\cite{VW} proves exponential correlation bounds for constant degrees, but decays quickly and becomes trivial at $d=O(\log n)$.
Our first result is that this is because of a good reason. Their technique is based on derivatives, and hence this fact extends to nonclassical polynomials. Moreover, it is tight for nonclassical polynomials. In the following, let $e:\T \to \C^*$ be defined as $e(x) = \exp(2 \pi i x)$.

\begin{thm}[Correlation bounds with modular sums for nonclassical polynomials (informal)]
\label{thm:intro:correlation}
Let $f:\F_2^n \to \T$ be a nonclassical polynomial of degree $d$. Then
$$
\E_{x \in \{0,1\}^n}\left[ e(f(x)) \omega_3^{x_1+\ldots+x_n} \right] \le 2^{-\Omega(n/4^d)}.
$$
Moreover, for any $\eps>0$ there exists a nonclassical polynomial $f:\F_2^n \to \T$ of degree $O(\log (n/\eps))$ such that
$$
\E_{x \in \{0,1\}^n}\left[ e(f(x)) \omega_3^{x_1+\ldots+x_n} \right] \ge 1-\eps.
$$
\end{thm}
So, the Viola-Wigderson technique is bounded for degrees smaller than $O(\log n)$, because it extends to nonclassical polynomials of that degree, for which it is tight. We note that the modulus $3$ in Theorem~\ref{thm:intro:correlation} can be replaced with any fixed odd modulus.

Another boolean function which was shown by Razborov and Smolensky~\cite{Raz2,Sm} to be hard for $\mathrm{AC}^0(\oplus)$ circuits is the majority function $\Maj:\F_2^n \to \F_2$. The proof relies on the following key fact. If $f:\F_2^n \to \F_2$ is a degree $d$ polynomial then
 \begin{equation}\label{eq:razsmo}
\Pr_{x \in \{0,1\}^n}\left[f(x)=\Maj(x)\right]  \leq  \frac{1}{2} + O\left( \frac{d}{\sqrt{n}}\right).
\end{equation}
Equivalently, this can be presented as a correlation bound
$$
\E_{x \in \{0,1\}^n}\l[(-1)^{f(x)} (-1)^{\Maj(x)}\r] \le O\left( \frac{d}{\sqrt{n}}\right).
$$
This is known to be tight for degree $d=1$ (as say $x_1$ has correlation $\Omega(1/\sqrt{n})$ with the majority function) and also for $d=\Omega(\sqrt{n})$, since there exist polynomials of that degree which approximate well the majority function, or any symmetric function for that matter. However, it is not known if these bounds are tight for degrees $1 \ll d \ll \sqrt{n}$. We study this question for nonclassical polynomials. We show that there are nonclassical polynomials of degree $O(\log n)$ with a constant correlation with the majority function.
\begin{thm}[Correlation bounds with majority for nonclassical polynomials (informal)]\label{thm:intro:majcor}
There exists a nonclassical polynomial $f:\F_2^n \to \T$ of degree $O(\log n)$ such that 
$$
\l|\E\l[e(f(x))(-1)^{\Maj(x)}\r]\r| \ge \Omega(1).
$$
\end{thm}
So, the Razborov-Smolensky technique separates classical from nonclassical polynomials, since classical polynomials of degree $O(\log n)$ have negligible correlation with the majority function, while as we show above, this is false for nonclassical polynomials.

\paragraph{Exact computation by polynomials.}
A related problem to correlation bounds is that of exact computation with good probability. For classical polynomials the two problems are equivalent, but this is not the case for nonclassical polynomials. Given a nonclassical polynomial $f:\F_2^n \to \T$,  we can ask what is the probability that $f$ is equal to a boolean function, say the majority function. To do so, we identify naturally $\F_2$ with $\{0,1/2\} \subset \T$, and consider $\Maj:\F_2^n \to \{0,1/2\}$.  We show the following result, which gives a partial answer to the question. 

\begin{thm}[Exact computation of majority by nonclassical polynomials (informal)]\label{thm:intro:majexact}
Let $f:\F_2^n \to \T$ be a nonclassical polynomial of degree $d$. Then,
$$
\Pr_{x \in \{0,1\}^n}[f(x)=\Maj(x)]\leq \frac{1}{2}+O\l(\frac{d 2^d}{\sqrt{n}}\r).
$$
\end{thm}

We believe that the bound is not tight, and that, unlike for correlation bounds, nonclassical polynomials should not be able to exactly compute boolean functions better than classical polynomials. Specifically, we ask the following problem.

\begin{open}
Let $f:\F_2^n \to \T$ be a nonclassical polynomial of degree $d$. Show that
$$
\Pr_{x \in \{0,1\}^n}[f(x)=\Maj(x)]\leq \frac{1}{2}+O\l(\frac{d}{\sqrt{n}}\r).
$$
\end{open}

\paragraph{Weak representation of the OR function.}
We next move to the problem of weak representation of the OR function. Let $p_1,\ldots,p_r$ be distinct primes and let $m=p_1 \ldots p_r$. The goal is to construct a low degree polynomial
$f \in \Z_m[x_1,\ldots,x_n]$ such that $f(0^n)=0$ but $f(x) \ne 0$ for all nonzero $x \in \{0,1\}^n$. Such polynomials stand
at the core of some of the best constructions of Ramsey graphs~\cite{FW, Gro, Gop}\footnote{The current record is due to \cite{BRSW} which uses different techniques.}  and locally decodable codes \cite{Yek, Efr, DGY, BDL, DH}, and were further investigated in~\cite{Sm, Bar, BeTa, BeGi,BBR, BT}. There are currently exponential gaps between the best constructions and lower bounds. Barrington, Beigel and Rudich \cite{BBR} showed that there exist polynomials of degree $O(n^{1/r})$ that weakly represent the OR function. The best lower bound is $\Omega(\log^{1/(r-1)} n)$, due to Barrington and Tardos \cite{BT}.

The definition of weak representation can be equivalently defined (via the Chinese Remainder Theorem) as follows. There exist polynomials $f_i:\F_{p_i}^n \to \F_{p_i}$ for $i=1,\ldots,r$ such that $f_1(0^n)=\ldots=f_r(0^n)=0$ but for any nonzero $x \in \{0,1\}^n$, there exists an $i$ for which $f_i(x) \ne 0$. This definition can be naturally extended to nonclassical polynomials, where we consider $f_i:\F_{p_i}^n \to \T$. We show that the Barrington-Tardos lower bound extends to nonclassical polynomials, and it is tight up to polynomial factors.

\begin{thm}[Weak representation of OR for nonclassical polynomials (informal)]
\label{thm:intro:weak}
Let $p_1,\ldots,p_r$ be distinct primes,
and $f_i:\F_{p_i}^n \to \T$ be nonclassical polynomials which weakly represent the OR function. Then
$$
\max \deg(f_i) \ge \Omega(\log^{1/r} n).
$$
Moreover, for any fixed prime $p$, there exists a nonclassical polynomial $f:\F_p^n \to \T$ of degree $O(\log n)$ which weakly represents the OR function.
\end{thm}

Thus, the proof technique of Barrington-Tardos cannot extend beyond degree $O(\log n)$, as it applies to nonclassical polynomials as well, for which the $O(\log n)$ bound holds even for prime modulus. We note that unlike in the case of Theorem~\ref{thm:intro:correlation}, where the lower bound proof of~\cite{VW} extended naturally to nonclassical polynomials, extending the lower bound technique of~\cite{BT} to nonclassical polynomials requires several nontrivial modifications of the original proof.

As an aside, in the classical setting, we present an improvement in the degree of a symmetric polynomial that weakly represents OR. This improves the result in \cite{BBR} in the growing modulus case and constructs a polynomial whose degree is modulus independent. For more details, see Appendix~\ref{app:or}.

\paragraph{Pseudorandom generators for low degree polynomials.}
Consider for simplicity polynomials over $\F_2$. A distribution $D$ over $\F_2^n$ is said to fool polynomials of degree $d$ with error $\eps$, if for any polynomial $f:\F_2^n \to \F_2$ of degree at most $d$, we have
$$
\left| \Pr_{x \sim D}[f(x)=0] - \Pr_{x \in \F_2^n}[f(x)=0] \right| \le \eps.
$$
Distributions which fool linear functions (e.g. $d=1$) are called small bias generators, and optimal constructions of them (up to polynomial factors) were given in~\cite{naor1993small,alon1992simple}, with seed length $O(\log n/\eps)$. A sequence of works~\cite{BV,lovett2009unconditional,viola2009sum} showed that small
bias generators can be combined to yield generators for larger degree polynomials. The best construction to date is by Viola~\cite{viola2009sum}, who showed that the sum of $d$ independent small bias generators with error approximately $\eps^{2^d}$ fools degree $d$ polynomials with error $\eps$. Thus, his construction has seed length $O(2^d \log(1/\eps) + d \log n)$, and becomes trivial for $d=\Omega(\log n)$. It is not clear whether it is necessary to require the small bias generators to have smaller error than the required error for the degree $d$ polynomials, and this is the main source for the loss in parameters when considering large degrees.

There is a natural extension of these definitions to nonclassical polynomials. If $f:\F_2^n \to \T$ is a nonclassical polynomial of degree $d$, then we require that
$$
\left| \E_{x \sim D}[e(f(x))] - \E_{x \in \F_2^n}[e(f(x))] \right| \le \eps.
$$
The proof technique of Viola is based on derivatives, and we note here (without proof) that it extends to nonclassical polynomials in a straightforward way. We suspect that it
is tight for nonclassical polynomials, however we were unable to show that. Thus, we raise the following open problem.

\begin{open}
Fix $\eps>0, d \ge 1$. Does there exist a small bias generator with error $\gg \eps^{2^d}$, such that the sum of $d$ independent copies of the generator does not fool degree $d$ nonclassical polynomials with error $\eps$?
\end{open}

\subsection{Organisation}
We start with some preliminaries in Section~\ref{sec:prelim}. In Section~\ref{sec:mod}, we prove the bounds on approximation of modular sums by nonclassical polynomials. Next, in Section~\ref{sec:maj}, we analyze the approximation of the majority function by nonclassical polynomials in the correlation model and the exact computation model. We prove the results on the weak representation of the OR function in Section~\ref{sec:or}. We describe in Appendix~\ref{app:or} an improvement in the degree of classical polynomials which weakly represent the OR function.

\paragraph{Acknowledgement.}We thank Parikshit Gopalan for fruitful discussions that led to the result on the classical OR representation in Appendix~\ref{app:or}. The first author would also like to thank his advisor, David Zuckerman, for his guidance and encouragement.

\section{Preliminaries}\label{sec:prelim}

Let $\N=\{1,2,\ldots\}$ denote the set of positive integers. For $n \in \N$, let $[n]:=\{1,2,\ldots , n\}$. Let $\T=\R/\Z$ denote the torus. This is an abelian group under addition. Let $e:\T  \to \C^*$ be defined by $e(x)=\exp(2\pi i x)$.

\paragraph{Nonclassical polynomials.}
Let $\F_p$ be a prime finite field. Given a function $f:\F_p^n \to \T$, its directional derivative in direction $h \in \F_p^n$ is $D_h f:\F_p^n \to T$, given by
$$
D_h f(x) = f(x+h) - f(x).
$$
Polynomials are defined as functions which are annihilated by repeated derivatives.

\begin{define}[Nonclassical polynomials]
A function $f:\F_p^n \to \T$ is a polynomial of degree at most $d$ if $D_{h_1} \ldots D_{h_{d+1}} f \equiv 0$ for any $h_1,\ldots,h_{d+1} \in \F_p^n$.
The degree of $f$ is the minimal $d$ for which this holds.
\end{define}

Classic polynomials satisfy this definition. Let $|\cdot|$ denote the natural map from $\F_p$ to $\{0,1,\ldots , p-1\} \subseteq \Z$. If $P:\F_p^n \to \F_p$ is a (standard) polynomial of degree $d$,
then $f(x) = |P(x)| / p \pmod{1}$ is a nonclassical polynomial of degree $d$. For degrees $d \le p$, it turns out that these are the only possible polynomials. However, when $d>p$, there are more polynomials
than just these arising from the classical ones, from which the term \emph{nonclassical polynomials} arise. A complete characterization of nonclassical polynomials was developed by Tao and Ziegler~\cite{TZ11}.
They showed that a function $f:\F^n \to \T$ is a polynomial of degree $\le d$ if and only if it has the following form:
$$
f(x_1,\ldots,x_n) = \alpha + \sum_{0 \le e_1,\ldots,e_n \le p-1, k \ge 0: \sum e_i + (p-1)k \le d} \frac{c_{e_1,\ldots,e_n,k} |x_1|^{e_1} \ldots |x_n|^{e_n}}{p^{k+1}} \pmod{1}.
$$
Here, $\alpha \in \T$ and $c_{e_1,\ldots,e_n,k} \in \{0,1,\ldots,p-1\}$ are uniquely determined. The coefficient $\alpha$ is called the \emph{shift} of $f$, and the largest $k$ for which $c_{e_1,\ldots,e_n,k} \ne 0$ for some $e_1,\ldots,e_n$ is called the \emph{depth} of $f$. Classical polynomials correspond to polynomials with $0$ shift and $0$ depth. In this work, we assume without loss of generality that all polynomials have shift $0$.
Define $\U_{p,k}:=\frac{1}{p^k} \Z / \Z$ which is a subgroup of $\T$. Then, the image of polynomials of depth $k-1$ lie in $\U_{p,k}$. We prove the following lemma which shows that nonclassical polynomials can be ``translated" to classical polynomials of a somewhat higher degree, at least if we restrict our attention to boolean inputs.

\begin{lem}\label{lemma:nonclassical_to_classical}
Let $f:\F_p^n \to \T$ be a polynomial of degree $d$ and depth $\le k-1$. Let $\varphi:\U_{p,k} \to \F_p$ be any function. Then there exists a classical polynomial $g:\F_p^n \to \F_p$ 
of degree at most $(p^k-1) d$, such that
$$
g(x) = \varphi(f(x)) \qquad \forall x \in \{0,1\}^n.
$$
\end{lem}

\begin{proof}
By the characterization of nonclassical polynomials, we have
$$
f(x) = \sum_{e,j} \frac{c_{e,j} |x_1|^{e_1} \ldots |x_n|^{e_n}}{p^j}
$$
where the sum is over $e=(e_1,\ldots,e_n)$ with $e_i \in \{0,\ldots,p-1\}$, $1 \le j \le k$ such that $\sum e_i + (p-1)(j-1) \le d$. We only care about the evaluation of $f$ on the boolean hypercube, which allows
for some simplifications. For any $x \in \{0,1\}^n$ we have $|x_1|^{e_1} \ldots |x_n|^{e_n} = \prod_{i \in I} x_i$ where $I = \{i: e_i \ne 0\}$. Thus, we can define an integer polynomial $P(x) = \sum_I c'_I \prod_{i \in I} x_i$ such that
$$
f(x) = \frac{P(x)}{p^k} \pmod{1} \qquad \forall x \in \{0,1\}^n,
$$
where $c'_I = \sum_{e: \{i:e_i \ne 0\}=I} \sum_j p^{k-j} c_{e,j}$. In particular, note that $P$ has degree at most $d$.  We may further simplify $P(x)=M_1(x)+\ldots+M_t(x)$, where each $M_i$ is a monomial of the form $\prod_{i \in I} x_i$, and monomials may be repeated (indeed, the monomial $\prod_{i \in I} x_i$ is repeated $c'_I$ times). Hence
$$
f(x) = \frac{M_1(x)+\ldots+M_t(x)}{p^k} \pmod{1} \qquad \forall x \in \{0,1\}^n.
$$
We care about the first $k$ digits in base $p$ of $P(x)=\sum M_i(x)$. These can be captured via the symmetric polynomials, using the fact that $M_i(x) \in \{0,1\}$ for all $x \in \{0,1\}^n$.

The $\ell$-th symmetric polynomial in $z=(z_1,\ldots,z_t)$, for $1 \le \ell \le t$, is a classical polynomial of degree $\ell$ defined as
$$
S_{\ell}(z)=\sum_{S \subset [t], |S|=\ell}\prod_{i \in S} z_{i}.
$$
When $z \in \{0,1\}^t$, it follows by Lucas theorem~\cite{Lucas} that the $i$-th digit of $z_1+\ldots+z_t$ in base $p$ is given by $S_{p^i}(z) \pmod{p}$. 

So, define a polynomial $Q:\F_p^k \to \F_p$ such that $Q(a_0,\ldots,a_{k-1}) = \varphi(\sum a_i p^i / p^k)$ for all $a_0,\ldots,a_{k-1} \in \{0,\ldots,p-1\}$, and polynomials $R_i:\F_p^n \to \F_p$ for $i=0,\ldots,k-1$ by $R_i(x)=S_{p^i}(M_1(x),\ldots,M_t(x))$. Note that $\deg(R_i) \le p^i d$. Define $g(x) = Q(R_0(x),\ldots,R_{k-1}(x))$. Then we have that
$$
g(x) = \varphi(f(x)) \qquad \forall x \in \{0,1\}^n.
$$
To conclude, we need to bound the degree of $g$. As monomials in $Q$ raise each variable to degree at most $p-1$, we 
have $\deg(g) \le (p-1) \sum \deg(R_i) \le (p^k-1) d$.
\end{proof}

\paragraph{Gowers uniformity norms.}
Let $F:\F^n \to \C$. The (multiplicative) derivative of $F$ in direction $h \in \F^n$ is given by $(\Delta_h F)(x) = F(x+h) \overline{F(x)}$. One can verify that if $f:\F^n \to \T$ and $F=e(f)$ then $\Delta_h F = e(D_h f)$.
The $d$-th Gowers uniformity norm $\|\cdot\|_{U^d}$ is defined as
$$
\|F\|_{U^d} := \left( \E_{h_1,\ldots,h_d,x \in \F^n}[\Delta_{h_1}\ldots \Delta_{h_d} F(x)] \right)^{1/2^d}.
$$
Observe that $\|F\|_{U^1} = | \E_x[F(x)] |$, which is a semi-norm. For $d \ge 2$, the Gowers uniformity norm turns out to indeed be a norm (but we will not need that).
The following lists the properties of the Gowers uniformity norm that we would need. For a proof and further details, see~\cite{Gow}.
\begin{itemize}
\item Let $f:\F^n \to \T$ and $F=e(f)$. Then $0 \le \|F\|_{U^d} \le 1$, where $\|F\|_{U^d}=1$ if and only if $f$ is a polynomial of degree $\le d-1$.
\item If $f:\F^n \to \T$ is a polynomial of degree $\le d-1$ then $\|F e(f)\|_{U^{d}} = \|F\|_{U^{d}}$ for any $F:\F^n \to \C$.
\item If $F(x_1,\ldots,x_n) = F_1(x_1) \ldots F_n(x_n)$ then $\|F\|_{U^d} = \|F_1\|_{U^d} \ldots \|F_n\|_{U_d}$.
\item (Gowers-Cauchy-Schwarz) For any $F:\F^n \to \C$ and any $d \ge 1$,
$$
0 \le \|F\|_{U^1} \le \|F\|_{U^2} \le \ldots \le \|F\|_{U^d}.
$$
\end{itemize}

\section{Approximating modular sums by polynomials}\label{sec:mod}

Viola and Wigderson~\cite{VW} proved that low-degree polynomials over $\F_2$ cannot correlate to the sum modulo $m$, as long as $m$ is odd.
Their proof technique is based on the Gowers uniformity norm. As such, it extends naturally to nonclassical polynomials. We capture that by the following theorem.
In the following, let $\omega_m = \exp(2 \pi i /m)$ be a primitive $m$-th root of unity.

\begin{thm}[Extension of \cite{VW} to nonclassical polynomials]\label{thm:modvw}
Let $f:\F_2^n \to \T$ be a polynomial of degree $<d$. Let $m \in \N$ be odd. Then for any $a \in \{1,\ldots,m-1\}$,
$$
\E_{x \in \{0,1\}^n} \left [e(f(x)) \cdot \omega_m^{a(x_1+\ldots+x_n)} \right] \le \exp(-c n/4^d)
$$
where $c=c_m>0$.
\end{thm}

\begin{proof}
Let $F(x) = e(f(x)) \cdot \omega_m^{a(x_1+\ldots+x_n)}$. By the properties of the Gowers uniformity norm,
$$
|\E_x[F(x)]| \le \|F\|_{U^d} = \|\omega_m^{a(x_1+\ldots+x_n)}\|_{U^d} = \prod_{i=1}^n \|\omega_m^{a x_i}\|_{U^d} = \|e(g)\|_{U^d}^n,
$$
where $g:\F_2 \to \T$ is given by $g(0)=0, g(1)=a/m$. A routine calculation shows that
$$
D_{h_1} \ldots D_{h_d} g(x) = \bigg \{
\begin{array}{ll}
a' / m&\textrm{if } h_1=\ldots=h_d=1, x=0\\
-a' / m&\textrm{if } h_1=\ldots=h_d=1, x=1\\
0&\textrm{otherwise}
\end{array}
$$
where $a' = a 2^{d-1}$ is nonzero modulo $m$. Hence $\|e(g)\|_{U^d}^{2^d}=(1-2^{-d}) + 2^{-d} \cos(2 \pi a'/m) \le 1 - 2^{-d} \cdot \Omega(1/m^2)$ and
$$
|\E[F]| \le \left(1 - 2^{-d} \cdot \Omega(1/m^2) \right)^{n/2^d} \le \exp(-c n/4^d)
$$
where $c = \Omega(1/m^2)$.
\end{proof}

This proof technique gives trivial bounds for $d \gg \log n$. Here, we show that this is for a good reason, as there are nonclassical polynomials of degree $O(\log n)$ which well approximate the sum modulo $m$.

\begin{thm}\label{thm:modvwcons}
Let $m \in \N$ be odd and fix $a \in \{1,\ldots,m-1\}$. For any $\eps>0$ there exists a polynomial $f:\F_2^n \to \T$ of degree $\log \l(\frac{n+m}{\eps}\r)+O(1)$ such that
$$
\E_{x \in \{0,1\}^n} \left [e(f(x)) \cdot \omega_m^{a(x_1+\ldots+x_n)} \right] = 1 + u
$$
where $|u| \le \eps$.
\end{thm}

\begin{proof}
Let $k \ge 1$ to be specific later. Let $r \in \{0,\ldots,m-1\}$ be such that $r \equiv a 2^k \pmod{m}$ and let $A=\frac{r - a 2^k}{m} \in \Z$.
Define $f:\F_2^n \to \T$ as
$$
f(x) = \frac{A(|x_1|+\ldots+|x_n|)}{2^k} \pmod{1}.
$$
Note that $f$ is a polynomial of degree $\le k$. For $x \in \{0,1\}^n$, if $x_1+\ldots+x_n=p m + q$ where $q \in \{0,\ldots,m-1\}$, then
$$
f(x)
\equiv \frac{A (pm+q)}{2^k}
\equiv \frac{rp+\frac{rq}{m}}{2^k}-\frac{aq}{m}=-\frac{aq}{m}+\theta_x \pmod{1},
$$
where $0 \le \theta_x \le (n+m)/2^k$. We choose $k\ge \log \l(\frac{n+m}{\eps}\r)+c$ for some universal constant $c$ so that $|e(\theta_x)-1| \le \eps$ for all $x$. Hence
$$
\left| \E \left [e(f(x)) \cdot \omega_m^{a(x_1+\ldots+x_n)} \right] - 1 \right| = \left| \E \left [e(\theta_x) - 1 \right] \right| \le \E \left [ \left| e(\theta_x) - 1 \right| \right] \le \eps.
$$
\end{proof}

\section{Approximating majority by nonclassical polynomials}\label{sec:maj}

The majority function $\Maj:\F_2^n \to \F_2$ is defined as
$$\Maj(x)=\l\{\begin{array}{ll}0 & \text{if }\sum_{i=1}^n |x_i| \leq n/2 \\
1 & \text{otherwise}\end{array}\r.
$$
We first show that is correlates well with a nonclassical polynomial of degree $O(\log n)$.

\begin{thm}There is a nonclassical polynomial $f:\F_2^n \to \T$ of degree $\log n+1$ such that $$\l|\E\l[(-1)^{\Maj(x)}e(f(x))\r]\r| \ge c,$$
where $c>0$ is an absolute constant.
\end{thm}
\begin{proof}
We assume $n$ even for the proof. The proof is similar for odd $n$. Let $A=\lfloor a \sqrt{n} \rfloor$ for $a>0$ to be specified later. Let $k$ be the smallest integer such that $2^k \geq n$. Set $$f(x)=\frac{A(\sum_{i=1}^n |x_i|-n/2)}{2^k}.$$
Note that $\deg(f)=\log n+1$. Now,
\begin{align*}
&\E\l[(-1)^{\Maj(x)}e(f(x))\r]\\
& =  2^{-n}\sum_{i=0}^{n/2} \binom{n}{i}e\l(A(i-n/2)/2^k\r)  -  2^{-n}\sum_{i=n/2+1}^n \binom{n}{i}e\l(A(i-n/2)/2^k\r)\\
&=2^{-n}\sum_{j=1}^{n/2} \binom{n}{n/2-j}e\l(-Aj/2^k\r)  -  2^{-n}\sum_{j=1}^{n/2} \binom{n}{n/2-j}e\l(Aj/2^k\r)+2^{-n}\binom{n}{n/2}\\
&=-2i \cdot 2^{-n}\sum_{j=1}^{n/2}\binom{n}{n/2-j}\sin \l(2 \pi Aj/2^k\r)+2^{-n}\binom{n}{n/2},
\end{align*}
where in the last equation $i=\sqrt{-1}$. Let $C=2^{-n}\sum_{j=1}^{n/2} \binom{n}{n/2-j}\sin \l(2 \pi Aj/2^k\r)$, so that $\l| \E\l[(-1)^{\Maj(x)}e(f(x))\r] \r| \geq 2C$. We will show that $C \ge \Omega(1)$.
Let $b>0$ be a constant to be specified later. We bound
$$
C \ge 2^{-n}\sum_{j=1}^{b\sqrt{n}}\binom{n}{n/2-j}\sin\l(2\pi Aj/2^k\r) - \exp(-2b^2),
$$
where the error term follows from the Chernoff bound. We set $a=1/8b$. For all $1 \le j \le b \sqrt{n}$ we have $2 \pi A j / 2^k \le \pi/4$. Applying the estimate $\sin(x) \ge x/2$
which holds for all $0 \le x \le \pi/4$, we obtain that
$$
C \ge \frac{\pi}{32 b \sqrt{n}} \cdot 2^{-n}\sum_{j=1}^{b\sqrt{n}}\binom{n}{n/2-j} j - \exp(-2b^2).
$$
Now, if $b$ is a large enough constant, standard bounds on the binomial coefficients give that
$$
2^{-n} \sum_{j=1}^{b\sqrt{n}}\binom{n}{n/2-j} j = \Omega(\sqrt{n}).
$$
Hence, we obtain that
$$
C \ge \Omega(1/b)-\exp(-2 b^2).
$$
If $b$ is chosen a large enough constant, this shows that $C \ge \Omega(1)$ as claimed.
\end{proof}

We next show that the Razborov-Smolensky technique generalizes to nonclassical polynomials when we require the polynomial to exactly compute $\Maj$.
Recall that we identify $\F_2$ with $\{0,1/2\} \subset \T$ and consider $\Maj:\F_2^n \to \{0,1/2\}$.

\begin{thm}Let $f:\F_2^n \to \T$ be a nonclassical polynomial of degree $d$ and depth $<k$. Then,
$$
\Pr_{x \in \{0,1\}^n}[f(x)=\Maj(x)]\leq \frac{1}{2}+O\l(\frac{2^k d}{\sqrt{n}}\r).
$$
\end{thm}

\begin{proof}
Let $\varphi:\U_{2,k} \to \F_2$ be defined as $\varphi(0)=0$, $\varphi(1/2)=1$ and choose arbitrarily $\varphi(x)$ for $x \in \U_{2,k} \setminus \{0,1/2\}$. Applying Lemma~\ref{lemma:nonclassical_to_classical}, there exists a classical polynomial $g:\F_2^n \to \F_2$ such that $g(x)=\varphi(f(x))$ for all $x \in \F_2^n$, where $\deg(g) \le (2^k-1)d$. In particular,
$$
\Pr_{x \in \F_2^n}[g(x)=\Maj(x)] \ge \Pr_{x \in \F_2^n}[f(x) = \Maj(x)].
$$
Hence, we can apply the Razborov-Smolensky~\cite{Raz2,Sm} bound to $g$ and conclude that
$$
\Pr[f(x) = \Maj(x)] \le \frac{1}{2} + O\l(\frac{\deg(g)}{\sqrt{n}}\r).
$$
\end{proof}

\section{Weak representation of the OR function}\label{sec:or}

A set of classical polynomials $f_i:\F_{p_i}^n \to \F_{p_i}$ is said to weakly represent the OR function if they all map $0^n$ to zero, and for any other point in the boolean hypercube, at least one of them map it to a nonzero value. This definition extends naturally to nonclassical polynomials.

\begin{define} Let $p_1,\ldots,p_r$ be distinct primes. A set of polynomials $f_i:\F_{p_i}^n \to \T$ weakly represent the OR function if
\begin{itemize}
\item $f_1(0^n)=\ldots=f_r(0^n)=0$.
\item For any $x \in \{0,1\}^n \setminus 0^n$, there exists some $i$ such that $f_i(x) \ne 0$.
\end{itemize}
\end{define}

It is well known that a single classical polynomial $f:\F_p^n \to \F_p$ which weakly represents the OR function, must have degree at least $n/(p-1)$. This is since $f(x)^{p-1}$ computes the OR function on $\{0,1\}^n$, and hence its multi-linearization (obtained by replacing any power $x_i^{e_i}$, $e_i \ge 1$ with $x_i$) must be the unique multi-linear extension of the OR function, which has degree $n$.

We first show that there is a nonclassical polynomial of degree $O(\log n)$ which weakly represents the OR function.

\begin{lem}
There exists a polynomial $f:\F_p^n \to \T$ of degree $O(p \lceil \log_p n \rceil)$ which weakly represents the OR function.
\end{lem}

\begin{proof}
Let $k \ge 1$ be minimal such that $p^k>n$. Define $f(x)=\frac{|x_1|+\ldots+|x_n|}{p^k}$. This is a polynomial of degree $1 + (p-1)(k-1)$. Clearly $f(0^n)=0$ and $f(x) \ne 0$ for any $x \in \{0,1\}^n \setminus 0^n$.
\end{proof}

We show that allowing for multiple nonclassical polynomials can only improve this simple construction by a polynomial factor.

\begin{thm}\label{thm:weakor}
Let $p_1,\ldots,p_r$ be distinct primes, and let $p=\max(p_1,\ldots,p_r)$. Let $f_i:\F_{p_i}^n \to \T$ be polynomials which weakly represent the OR function. Then at least one of the polynomials
must have degree $\Omega((\log_p{n})^{1/r})$.
\end{thm}

The proof is an adaptation of the result of Barrington and Tardos~\cite{BT}, who proved similar lower bounds for classical polynomials. We start by showing that a low degree polynomial $f$ with $f(0)=0$ must have
another point $x$ with $f(x)=0$.

\begin{clm}\label{clm:singlepoly}
Let $f:\F_p^n \to \T$ be a polynomial of degree $d$ and depth $\le k-1$ such that $f(0)=0$. If $n > (p^k-1) d$ then there exists $x \in \{0,1\}^n \setminus 0^n$ such that $f(x) = 0$.
\end{clm}

We note that the bound on $n$ is fairly tight, as $f(x)=(x_1+\ldots+x_n)/p^k \pmod{1}$ violates the conclusion of the claim whenever $n<p^k$.

\begin{proof}
Let $\varphi:\U_{p,k} \to \F_p$ be given by $\varphi(0)=0$, $\varphi(x)=1$ for all $x \ne 0$. Applying Lemma~\ref{lemma:nonclassical_to_classical}, 
there exists a classical polynomial $g:\F_p^n \to \F_p$ of degree $\le (p^k-1) d$ such that $g(x)=0$ if $f(x)=0$, and $g(x)=1$ if $f(x) \ne 0$, for all $x \in \{0,1\}^n$.
If $f(0^n)=0$ but $f(x) \ne 0$ for all nonzero $x \in \{0,1\}^n$, then $g$ computes the OR function over $\{0,1\}^n$. Hence, $\deg(g) \ge n$, which leads to 
a contradiction whenever $n>(p^k-1)d$.
\end{proof}

We next extend Claim~\ref{clm:singlepoly} to a find a common root for a number of polynomials.

\begin{clm}\label{clm:manypoly}
Let $f_1,\ldots f_{r}:\F_p^n \to \T$ be polynomials of degree $d$ and depth $\le k-1$ such that $f_i(0)=0$ for all $i \in [r]$.
If $n > (p^k-1) d r$  then there exists $x \in \{0,1\}^n \setminus 0^n$ such that $f_i(x)=0$ for all $i \in [r]$.
\end{clm}

\begin{proof}
We construct an interpolating polynomial for $f_1,\ldots,f_{r}$. Following the proof of Claim~\ref{clm:singlepoly}, for each $f_i$ there exists a classical polynomial $g_i:\F_p^n \to \F_p$ satisfying the following. For any $x \in \{0,1\}^n$, if $f_i(x)=0$ then $g_i(x)=0$, and if $f_i(x) \ne 0$ then $g_i(x)=1$. Moreover, $\deg(g_i) \leq (p^k-1) d$. Define $g:\F_p^n \to \F_p$ as $$g(x)=1-\prod_{i=1}^r(1-g_i(x)).$$
Note that $\deg(g) \le \sum \deg(g_i) \le (p^k-1) d r$. Suppose for contradiction that for every $x \in \{0,1\}^n \setminus 0^n$ there is an $i \in [r]$ such that $f_i(x) \neq 0$. Then $g(0)=0$ as $f_i(0)=0$ for all $i \in [r]$, but $g(x)=1$ for all
$x \in \{0,1\}^n \setminus 0^n$. Then $g$ computes the OR function over $\{0,1\}^n$, and hence $\deg(g) \ge n$. This leads to a contradiction whenever $n > (p^k-1) d r$.
\end{proof}

Next, we argue that the hamming ball of radius $d$ is an interpolating set for polynomials of degree $d$ over $\{0,1\}^n$. In the following, let $B(n,d) = \{x \in \{0,1\}^n: \sum x_i \le d\}$.

\begin{clm}\label{clm:interpolate}
Let $f:\F_p^n \to \T$ be a polynomial of degree $d$ such that $f(x)=0$ for all $x \in B(n,d)$. Then $f(x)=0$ for all $x \in \{0,1\}^n$.
\end{clm}

\begin{proof}
Towards contradiction, let $x^* \in \{0,1\}^n$ be a point such that $f(x^*) \ne 0$, with a minimal hamming weight. By assumption, the hamming weight of $x^*$ is at least $d+1$.
Let $i_1,\ldots,i_{d+1} \in [n]$ be distinct coordinates such that $x^*_{i_1}=\ldots=x^*_{i_{d+1}}=1$. Let $e_j \in \{0,1\}^n$ be the $j$-th unit vector, defined as $(e_j)_j=1$ and $(e_j)_{j'}=0$ for $j'  \ne j$. Define vectors $h_1,\ldots,h_{d+1} \in \F_p^n$ by $h_j = -e_{i_j}$.
Since $f$ is a degree $d$ polynomial, we have
$$
D_{h_1} \ldots D_{h_{d+1}} f \equiv 0.
$$
Evaluating this on $x^*$ gives
$$
\sum_{I \subset \{i_1,\ldots,i_{d+1}\}} (-1)^{|I|} f(x^* - \sum_{i \in I} e_i)=0.
$$
However, as we chose $x^*$ with minimal hamming weight such that $f(x^*) \ne 0$, we have $f(x^* - \sum_{i \in I} e_i)=0$ for all nonempty $I$. Hence also $f(x^*)=0$.
\end{proof}

Next, we prove that low degree polynomials must be zero on a large combinatorial box. In the following, we identify subsets $S \subset [n]$ with their indicator in $\{0,1\}^n$.
\begin{lem}\label{lem:lowerboundlem}
Let $f:\F_p^n \to \T$ be a polynomial of degree $d$ and depth $\le k-1$ such that $f(0)=0$. For $\ell \ge 1$, if
$n \ge 2dp^k\ell^{d+1}$
then there exist pairwise disjoint and nonempty sets of variables $S_1,\ldots ,S_{\ell} \subset [n]$ such that
$$
f\left(\sum_{i=1}^{\ell} y_i S_i\right)=0 \qquad \forall y \in \{0,1\}^{\ell}.
$$
\end{lem}

\begin{proof}
Fix $a_1,\ldots,a_{\ell}$ to be determined later such that $n \ge a_1+\ldots+a_{\ell}$. Let $A_1,\ldots,A_{\ell} \subset [n]$ be disjoint subsets of variables of size $|A_i|=a_i$. We will
find subsets $S_i \subset A_i$ such that $f(\sum y_i S_i)=0$ for all $y \in \{0,1\}^{\ell}$. As we may set the variables outside $A_1,\ldots,A_{\ell}$ to zero, we assume from now on
that $n=a_1+\ldots+a_{\ell}$.

First, set $a_1=p^k d$.  Consider the restriction of $f$ to $A_1$ by setting the remaining variables to zero. By Claim~\ref{clm:singlepoly}, there exists a nonempty set $S_1 \subset A_1$ such that $f(S_1)=0$.

Next, suppose that we already constructed $S_1 \subset A_1,\ldots,S_j \subset A_j$ for some $1 \le j<\ell$, such that $f(\sum y_i S_i)=0$ for all $y \in \{0,1\}^j$. For each $y \in \{0,1\}^j$, define a polynomial $f_y:\F_p^{A_{j+1}} \to \T$ by
$$
f_y(x') = f\left(\sum_{i=1}^j y_i S_i + x'\right)
$$
where $x'  \in \F_p^{A_{j+1}}$ denotes the variables in $A_{j+1}$. We will find a common nonzero root for ${f_y(x')}$.

First, consider only $y \in B(j,d)$. The number of such polynomials is $r = {j \choose \le d} = \sum_{i=0}^d {j \choose i}$.
Applying claim~\ref{clm:manypoly}, we have that if we choose $a_{j+1} \ge drp^k$ then there exists $S_{j+1} \subset A_{j+1}$ such that
$$
f_y\left(S_{j+1}\right) = 0 \qquad \forall y \in B(j,d).
$$
We claim that this implies that $f_y(S_{j+1})=0$ for all $y \in \{0,1\}^j$. To see that, define $g:\F_p^j \to \T$ by
$$
g(y) = f\left(\sum_{i=1}^j y_i S_i + S_{j+1}\right).
$$
This a polynomial of degree $d$, and by Claim~\ref{clm:interpolate}, if it is zero for all $y \in B(j,d)$, then it is the zero on all $\{0,1\}^d$. Hence, we have that $f(\sum_{i=1}^{j+1} y_i S_i)=0$ for all $y \in \{0,1\}^{j+1}$.

We now calculate the parameters. We have ${j \choose \le d} \le 2 j^d$, and hence it suffices to take $a_{j+1} = 2dj^dp^k$. Hence, we need $n \ge n_0$ for
$$
n_0 = \sum_{j=1}^{\ell} a_j  \leq 2dp^k \sum_{j=1}^{\ell} j^{d} \le 2dp^k\ell^{d+1}.
$$
\end{proof}

We are now ready to prove Theorem~\ref{thm:weakor}.

\begin{proof}[Proof of Theorem~\ref{thm:weakor}]
Let $p_1,\ldots,p_r$ be distinct primes, and let $p=\max(p_1,\ldots,p_r)$. Let $f_i:\F_{p_i}^n \to \T$ be polynomials of degree at most $d$ and depth at most $k-1$ which weakly represent the OR function. We fix integers $n \geq \ell_0=n_0 \geq \ell_1 \ldots \geq \ell_{r-1} \geq \ell_r=1$ which will be specified later. Applying Lemma~\ref{lem:lowerboundlem} to $f_1$ with parameter $\ell_1$, we get that as long as $n$ is large enough, we can find disjoint nonempty subsets $S_{1,1},\ldots,S_{1,\ell_1} \subset [n]$ such that $f_1(\sum y_i S_{1,i})=0$ for all $y \in \{0,1\}^{\ell_1}$.

Next, consider the restriction of $f_2$ to the combinatorial cube formed by $\{S_{1,i}\}$. That is, define $f'_2:\F_p^{\ell_1} \to \T$ by $f'_2(y) = f_2(\sum y_i S_{1,i})$. Note that $f'_2$ is a polynomial of degree at most $d$ and depth at most $k-1$. Applying Lemma~\ref{lem:lowerboundlem} to $f'_2$ with parameter $\ell_2$, we get that as long as $\ell_1$ is large enough, we can find disjoint nonempty subsets $S'_{2,1},\ldots,S'_{2,\ell_2} \subset [\ell_1]$ such that $f'_2(\sum y_i S'_{2,i})=0$ for all $y \in \{0,1\}^{\ell_2}$. Define $S_{2,1},\ldots,S_{2,\ell_2} \subset [n]$ by $S_{2,i} = \cup_{j \in S'_{2,i}} S_{1,j}$. Then $S_{2,1},\ldots,S_{2,\ell_2}$ are disjoint nonempty subsets of $[n]$, such that
$$
f_1\left(\sum_{i=1}^{\ell_2} y_i S_{2,i}\right) = f_2\left(\sum_{i=1}^{\ell_2} y_i S_{2,i}\right) = 0 \qquad \forall y \in \{0,1\}^{\ell_2}.
$$

Continuing in this fashion, we ultimately find disjoint nonempty subsets $S_{r,1},\ldots,S_{r,\ell_r} \subset [n]$ such that
$$
f_1\left(\sum_{i=1}^{\ell_r} y_i S_{r,i}\right) = \ldots = f_r\left(\sum_{i=1}^{\ell_r} y_i S_{r,i}\right) = 0 \qquad \forall y \in \{0,1\}^{\ell_r}.
$$
In particular, $f_1,\ldots,f_r$ cannot weakly represent the OR function. This argument requires that for each $0 \leq i \leq r-1$, $\ell_{i} \ge 2dp^k\ell_{i+1}^{d+1}$, which can be
satisfied if
$$
n \ge n_0 = (2dp^k)^{(d+1)^{r-1}}.
$$
Now, $k \le d/(p-1) + 1$ and hence $p^k \le p^{d/(p-1)+1} \le 2^d p$. As we can trivially bound $2d \le 2^{d}$ we obtain the simplified bound
$$
n_0 \le 2^{4(d+1)^r \cdot \log p}.
$$
Thus, if $f_1,\ldots,f_r$ do weakly represent the OR function, at least one of the must have degree $d \ge \Omega((\log_p{n})^{1/r})$.
\end{proof}

\bibliographystyle{alpha}
\bibliography{correlation}

\appendix

\section{Improved weak OR representation by classical polynomials}\label{app:or}

In this section, we construct a low degree polynomial over $\Z_m$ that weakly represents the OR function. Recall that the task is to construct a polynomial $P$ in $\Z_m[x_1,\ldots ,x_n]$ such that $P(0)=0$ and $P(x) \neq 0$ for any nonzero $x \in \{0,1\}^n$. Let $m=p_1,\ldots ,p_r$ for pairwise distinct primes $p_i$. Let $\ell(m)$ be the largest prime divisor of $m$. As mentioned before, the best result is due to Barrington, Beigel and Rudich \cite{BBR}, who constructed a symmetric polynomial of degree $O\l(\ell(m)n^{1/r}\r)$ that weakly represents the OR function. It is also well known \cite{BBR}, by Lucas' theorem that for symmetric functions, $d=\Omega\l(\ell(m)^{-1}n^{1/r}\r)$.

Our construction takes us closer to the lower bound. We construct symmetric polynomials that have modulus independent degree, that is, $d=O\l(n^{1/r}\r)$.

\begin{thm}Let $m=\prod_{i=1}^rp_i$ for pairwise distinct primes $p_i$. Then there exists an explicit polynomial $P \in \Z_m[x_1,\ldots ,x_n]$ of degree at most $2\lceil n^{1/r}\rceil$ such that $P$ weakly represents $OR$ modulo $m$.  \end{thm}
\begin{proof}
For each $1 \leq i \leq r$, let $e_i$ be the smallest integer such that $p_i^{e_i}>\l\lceil n^{1/r}\r\rceil$.

\paragraph{The construction.} Let $S_j$ be the $j$-th symmetric polynomial in $x=(x_1,\ldots ,x_n)$. Let $q_i$ be a quadratic non residue in $\Z_{p_i}$ for odd $p_i$.  Define $P \in \Z[x_1,\ldots ,x_n]$ as follows. Let $$P(x)=z_{i1}^2-q_iz_{i2}^2 \mod{p_i}, \text{ for odd } p_i,$$ and $$P(x)=z_{i1}^2+z_{i1}z_{i2}+z_{i2}^2 \mod{p_i}, \text{ for } p_i=2,$$ where $$z_{i1}=1-\prod_{j=0}^{e_i-2}(1-S_{p_i^j}(x)^{p_i-1})$$ and $$z_{i2}=s_{p_i^{e_i-1}}(x).$$ This uniquely defines $P(x) \mod{m}$.

Note that $P(x)=0 \mod{p_i}$ if and only if $z_{i1}=z_{i2}=0 \mod{p_i}$. This follows from the irreducibility of $x^2-q_i$ over $\Z_{p_i}$ for odd $p_i$ and $x^2+x+1$ over $\Z_2$.

If $x=0$, then $z_{1i}=z_{2i}=0 \mod{p_i}$ for all $i$ and hence $P(x)=0 \mod{m}$.

Let $\wt(x):=\sum_{i=1}^n |x_i|$. Now, given $x \neq 0$, we have $\wt(x) \neq 0$. Therefore, $\wt(x) \neq 0 \mod{n+1}$. Thus, there exists $i_0$ such that $\wt(x) \neq 0 \mod{p_{i_0}^{e_{i_0}}}$. From here on, we set $p:=p_{i_0}, e:=e_{i_0}$. Consider the $p$-ary expansion of $\wt(x)$. Let $\wt(x)=\sum_{j=0}^{e-1}a_jp^j+tp^e$, $0 \leq a_i \leq p-1$. Since $\wt(x) \neq 0 \mod{p^e}$, we have for some $j$, $a_j \neq 0$.

We first note that since $x \in \{0,1\}^n$, we have $S_{p^j}(x)={\wt(x) \choose p^j}$. Therefore, by Lucas' theorem, we have $a_j=S_{p^j}(x) \mod{p}$.

Let $z_1=z_{i_01}, z_1=z_{i_02}$. Now, if $a_{e-1}\neq 0$, then $S_{p^{e-1}}(x)=z_2 \neq 0 \mod{p}$ and thus $P(x) \neq 0 \mod{p}$. Therefore, $P(x) \neq 0 \mod{m}$ and we are done. If on the other hand, if any $a_j \neq 0$ ($j \leq e-2$), then $S_{p^j}(x)\neq 0$. Thus, $z_1=1$ and hence $P(x) \neq 0 \mod{p}$. Therefore, $P(x) \neq 0 \mod{m}$.

Finally, we bound the degree of $P(x)$. The degree of each $z_{i1}$ is at most $(p_i-1)\sum_{j=0}^{e-2}p_i^j=p_i^{e_i-1}-1$. The degree of each $z_{i2}$ is $p_i^{e_i-1}$. Therefore the degree of $P(x)$ is $\max_{i}2p_i^{e_i-1}$. (Note that is where we improve on \cite{BBR}. Their upper bound is $p_i^{e_i}$.) Since $e_i$ is the least integer such that $p_i^{e_i}>\l\lceil n^{1/r}\r\rceil$, we have $p_i^{e_i} \leq p_i \l\lceil n^{1/r}\r\rceil$. Therefore, $p_i^{e_i-1} \leq \l\lceil n^{1/r}\r\rceil$ and this proves the theorem.

\end{proof}





\end{document}